\documentclass{ws-dmaa}
\usepackage{amsmath}
\usepackage{amssymb}

\usepackage{wasysym}
\usepackage{float}

\usepackage{subfigure}

\newtheorem{thm}{Theorem}
\newtheorem{defi}{Definition}

\newtheorem{lem}[thm]{Lemma}

\usepackage[caption=false,font=footnotesize]{subfig}
\usepackage{fixltx2e}

\newcommand{\DN}{\gamma}
\newcommand{\TDN}{\gamma_t}
\newcommand{\OT}{[1,2]}
\newcommand{\OTDN}{\gamma_{\OT}}
\newcommand{\OTTDN}{\gamma_{t\OT}}
\newcommand{\M}{{\cal M}_p}

\begin{document}

\markboth{F. Beggas, V. Turau, M. Haddad and H. Kheddouci}
{[1,2]-Domination in Generalized Petersen Graphs}

\catchline{}{}{}{}{}
%
%%%%%%%%%%%%%%%%%%%%%%%%%%%%%%%%%%%%%%%%%%%%%%%%%%%%%%%%%%%%%%%%%%%%

\title{[1,2]-DOMINATION IN GENERALIZED PETERSEN GRAPHS\footnote{This work is supported by PHC PROCOPE 2015-2017, project id 33394TD and PPP PROCOPE 2015, project id 57134870.}
}

\author{FIROUZ BEGGAS$^a$, VOLKER TURAU$^b$, MOHAMMED HADDAD$^a$\footnote{Contact author: mohammed.haddad@univ-lyon1.fr} \\ HAMAMACHE KHEDDOUCI$^a$}

\address{$^a$University of Claude Bernard Lyon 1, 43 Bd du 11 Novembre 1918, F-69622, Villeurbanne, France\\ $^b$Institute of Telematics, Hamburg University of Technology, 21073 Hamburg, Germany, Am Schwarzenberg-Campus 3}

\maketitle

%\begin{history}
%\received{Day Month Year}
%\revised{Day Month Year}
%\accepted{Day Month Year}
%\published{Day Month Year}
%\end{history}

\begin{abstract}
A vertex subset $S$ of a graph $G=(V,E)$ is a $\OT$-dominating set
if each vertex of $V\backslash S$ is adjacent to either one or two
vertices in $S$. The minimum cardinality of a $\OT$-dominating set
of $G$, denoted by $\OTDN(G)$, is called the $[1,2]$-domination
number of $G$. In this paper the $[1,2]$-domination and the
$[1,2]$-total domination numbers of the generalized Petersen graphs
$P(n,2)$ are determined.
\end{abstract}


%\ccode{Mathematics Subject Classification: 11xxx, 11xxx, 11xxx}

\section{Introduction}
The study of domination problems in graph theory has a long history.
For an undirected graph $G=(V,E)$ a subset $S\subseteq V$ is a {\em
  dominating set} if every vertex not in $S$ has a neighbor in $S$.
The {\em domination number} $\DN(G)$ is the minimum size of a
dominating set in $G$. For many classes of graphs the exact values of
$\DN(G)$ are known, e.g., $\DN(P_n) =\DN(C_n) =\lceil n/3\rceil$. Here
$P_n$ and $C_n$ are the paths and cycle graphs respectively with $n$
vertices. For the class of generalized Petersen graphs $P(n,2)$
introduced by Watkins \cite{Watkins:1969} it was conjectured by Behzad
et al.\ that $\DN(P(n,2)) = \lceil 3n/5 \rceil$ holds
\cite{Behzad:2008}. This conjecture was later independently verified
by several researchers \cite{Ebrahimi:2009,Fu:2009,Yan:2009}.

Over the years different variations of graph domination were
introduced, e.g., connected domination, independent domination, and
total domination. The domination number $\DN(G)$ and the total
domination number $\TDN(G)$ of graph $G$ are among the most well
studied parameters in graph theory. Some of these domination numbers
are known for generalized Petersen graphs. Cao et al.\ computed the
total domination number of $P(n,2)$ as $\TDN(P(n,2)) = 2\lceil n/3
\rceil$ \cite{Cao:2009}. Further results can be found in
\cite{Li:2013,Zelinka2002}.

This paper considers
$\OT$-domination, a concept introduced by Chellali et al.\
\cite{Chellali:2013}. A subset $S\subseteq V$ is a {\em
  $\OT$-dominating set} if every vertex not in $S$ has at least one
and at most two neighbors in $S$, i.e., $1\le |N(v)\cap S|\le 2$ for
all $v\in V \setminus S$. The {\em $\OT$-domination number} $\OTDN(G)$
is the minimum size of a $\OT$-dominating set in $G$. Obviously
$\DN(G) \le \OTDN(G)$ for any graph $G$. Chellali et al.\ proved that
if $G$ is a $P_4$-free graph then $\DN(G) = \OTDN(G)$. A
characterization of graphs with this property is an open problem. More
results about $\OT$-domination can be found in \cite{Yang:2014}.
 
This paper also deals with the $\OT$-total domination defined as follow.
A subset $S\subseteq V$ is a {\em
  $\OT$-total dominating set} if every vertex $v$ in $V$ has at least one
and at most two neighbors in $S$, i.e., $1\le |N(v)\cap S|\le 2$ for
all $v\in V$. The {\em $\OT$-total domination number} $\OTTDN(G)$
is the minimum size of a $\OT$-total dominating set in $G$.
Clearly, $\OTDN(G) \leq \OTTDN(G)$ for each graph $G$.

In this paper we analyze the $\OT$-domination numbers of the generalized
Petersen graphs $P(n,2)$ and prove the following theorem.

\begin{thm}
   $\OTDN(P(n,2)) = \left\{
    \begin{array}{ll}
        2n/3  & \mbox{if } n \equiv 0,3[6]  \\
        2\lfloor n/3 \rfloor+1 & \mbox{if } n \equiv 1[6]  \\
       2\lfloor n/3 \rfloor+2 & \mbox{otherwise.} 
    \end{array}
\right.$ for $n \ge 5$.   \label{thm:thm1}
\end{thm}

Note that $\OTDN(P(n,2))$ is by a factor of $10/9$ larger than
$\DN(P(n,2))$. After that, we investigate the problem of $\OT$-total
domination and prove the following result.

\begin{thm}
$\OTTDN(P(n,2)) = \left\{
\begin{array}{ll}
5  & \mbox{if } n = 5 \\
2n/3  & \mbox{if } n \equiv 0,3[6]  \\
2\lfloor n/3 \rfloor+2 & \mbox{otherwise.} 
\end{array}
\right.$ for $n \ge 6$.   \label{thm:thm2}
\end{thm}

Note that $\OTTDN(P(n,2))= \OTDN(P(n,2))$ except for the case $n=5$
and $ n \equiv 1[6]$. Surprisingly $\OTTDN(P(n,2))$ is almost equal to
$\TDN(P(n,2))$.

\section{Notation}
\label{sec:notation}
This paper uses standard notation from graph theory which can be found
in textbooks on graph theory such as \cite{Diestel:2012}. For an
extended study about domination concepts the reader is referred to
\cite{Haynes:1998}.

\begin{defi}
  Let $n, k\in \mathbb{N}$ with $k < n/2$. The {\em generalized
    Petersen graph} $P(n,k)$ is the undirected graph with vertices
  $\{u_0,\ldots,u_{n-1}\}\, \cup\, \{v_0,\ldots, v_{n-1}\}$ and edges
  $\{(u_i,u_{i+1}),(u_i,v_i), (v_i,v_{i+k}) \mid 0\le i < n\}$.
\end{defi}

The
graphs $P(n,k)$ are regular graphs with $2n$ vertices and $\Delta=3$.
The domination number $\DN(P(n,k))$ for some values of $k$ are known
\cite{Behzad:2008,Zelinka2002}. In particular Ebrahimi et al.\ proved
in \cite{Ebrahimi:2009} that $\DN(P(n,2)) = \lceil\frac{3n}{5}\rceil$.

\begin{figure}[htb]
\begin{center}
\includegraphics[scale=0.5]{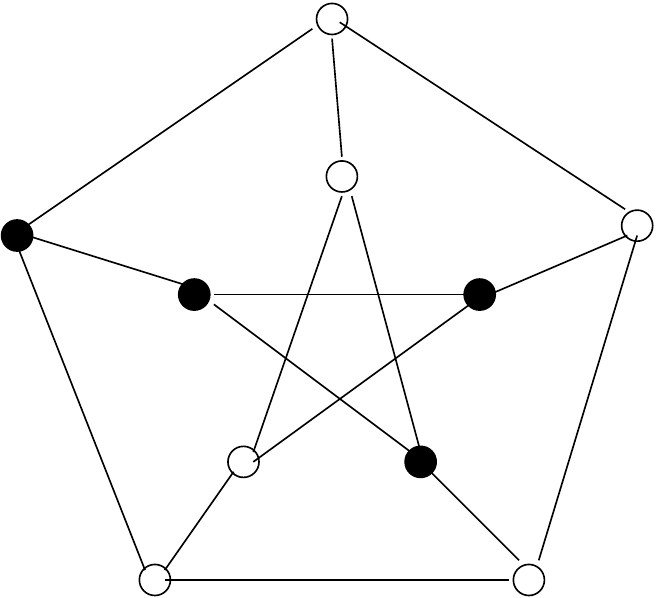}
\hspace*{1.5cm}
\includegraphics[scale=0.5]{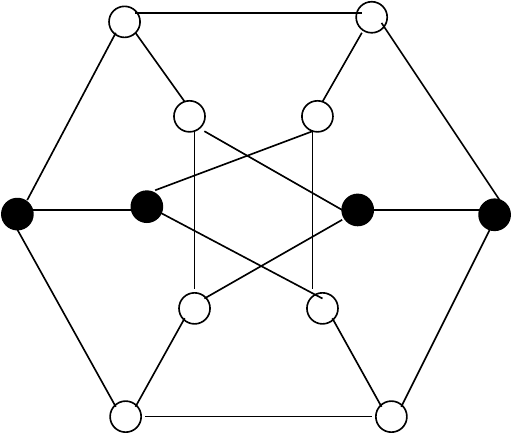}
\hspace*{1.5cm}
\includegraphics[scale=0.5]{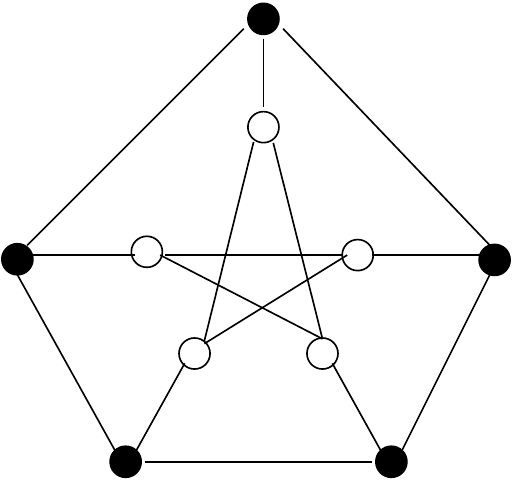}
  \end{center}
\label{fig:graph56}
  \caption{ The minimum $\OT$-domination sets of the generalized Petersen graphs $P(5,2)$ and $P(6,2)$
  and the minimum $[1,2]$-total dominating set for $P(5,2)$.}
\end{figure}

In this paper indices are always interpreted modulo $n$, e.g.
$v_{n+i}=v_i$. Fig.~$1$ shows the graphs
$P(5,2)$ and $P(6,2)$, vertices depicted in black form a \OT-dominating
set of minimum size, i.e., $\OTDN(P(5,2)) = \OTDN(P(6,2)) = 4$
and also for the graph $P(5,2)$, vertices 
depicted in black form a \OT-total dominating
set of minimum size $\OTTDN(P(5,2))=5$.

The proofs of this paper use the following notion of a {\em block}.

\begin{defi}
  A {\em block} $b$ of $P(n,2)$ is the subgraph induced by the six
  vertices $\{v_{i-1},v_{i},v_{i+1},u_{i-1},u_{i},u_{i+1}\}$ for any
  $i\in \{0,\ldots n-1\}$. A block is called {\em positive} if two
  of the indices of $\{v_{i-1},v_{i},v_{i+1}\}$ are odd, otherwise it
  is called {\em negative}.
\end{defi}

Fig.~\ref{img:blocks} shows a series of blocks of $P(n,2)$. The
second block is {\em positive} while the other two are {\em negative}.
Note that blocks can overlap. If $b$ is a block, the block to the
left is denoted by $b^-$ and that to the right by $b^+$.

% If $n \equiv 0 (3)$ then $P(n,2)$ can be partitioned into a
% alternating sequence of positive and negative 6-blocks.

\begin{figure*}[h!]
\begin{center}
 $b^-$ \hspace{2cm}
$b$ \hspace{2cm}
$b^+$
\includegraphics{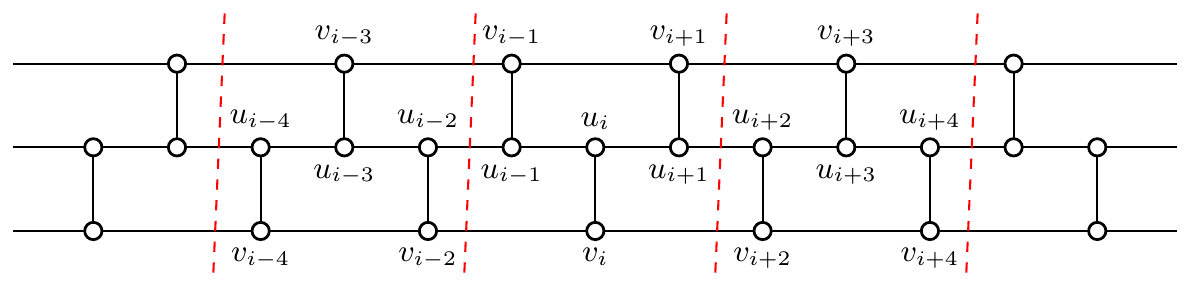} 
\end{center}
\label{img:blocks}
\caption{Partition of $P(n,2)$ into blocks.}
\end{figure*}

\begin{defi}
  Let $S$ be a $\OT$-dominating set. For a subset $U\subseteq V$
  denote by $\gamma_S(U)$ the number of vertices of $S$ that are in
  $U$, i.e., $\gamma_S(U)=|U \cap S|$. For $i\ge 0$ let ${\cal
    B}_i(S)$ be the set of all blocks $b$ with $\gamma_S(b)=i$.
\end{defi}

Note that ${\cal B}_0(S)= \emptyset$ for any dominating set
$S$ of $P(n,2)$. Denote by $f(n)$ the value of the right side of the
equation in Theorem~\ref{thm:thm1}. Note that $f(n) = f(n-6) + 4$ for
any $n\ge 5$.

\section{Determination of $\OTDN(P(n,2))$}
\label{sec:determ-otdnpn-2}
The correctness of Theorem~\ref{thm:thm1} for $n < 12$ can be
verified manually.
 \begin{lem}\label{lem:smallCases}
   $\OTDN(P(n,2)) = f(n)$ for $5\le n <12$. 
 \end{lem}

% By inspection, the graph $P(5,2)$ has
% $\gamma_{[1,2]}(P(5,2))=4$ (See Figure \ref{n=5}) then
% $\gamma_{[1,2]}(P(5,2))=f(5)$. For $n=6$, the graph $P(6,2)$ has
% $\gamma_{[1,2]}(P(6,2))=4=f(6)$ as depicted in Figure \ref{n=6}). By
% the same way, $\gamma_{[1,2]}(P(7,2))=5=f(7)$,
% $\gamma_{[1,2]}(P(8,2))=6=f(8)$, $\gamma_{[1,2]}(P(9,2))=6=f(9)$,
% $\gamma_{[1,2]}(P(10,2))=7=f(10)$ and
% $\gamma_{[1,2]}(P(11,2))=7=f(11)$.

\begin{proof}
By inspection, we can easily see that
the following sets $S_n$ are minimum $\OT$-dominating sets of $P(n,2)$.
$S_5= \{ u_1,v_1,v_3,v_4 \}$, $S_6= \{ u_1,v_1,u_4,v_4 \}$,
$S_7= \{ u_0,v_1,v_2,v_3,u_4 \}$, $S_8= \{ u_1,v_1,u_4,v_4,v_6,v_7 \}$, 
$S_9= \{ u_1,v_1,u_4,v_4,u_7,v_7 \}$, $S_{10}= \{ u_1,v_1,u_4,v_4,u_7,v_7,u_8,v_8 \}$ and
$S_{11}= \{ u_1,v_1,u_4,v_4,u_7,v_7,v_9,v_{10} \}$.
\end{proof}

%{\bf Should we describe the construction? Should be easy!}
\begin{lem}\label{lem:upperBound}
  $\OTDN(P(n,2)) \le f(n)$ for $n\ge 5$.
\end{lem}

\begin{proof}
  To prove that $f(n)$ is an upper bound of $\OTDN(P(n,2))$, we give
  in Fig.~\ref{fig:upperBound} the corresponding construction for each
  case. For $n \equiv 0,3[6]$, we choose the middle pair of nodes of
  each block. For the cases $n \equiv 2,4,5[6]$, we do the same as the
  previous case by choosing the middle pair of nodes of each block.
  Then, we add two dominating nodes as depicted in red in
  Fig.~\ref{fig:upperBound}. For the case $n \equiv 1[6]$, we choose
  two nodes from each block as shown in Fig.~\ref{fig:upperBound}
  except in the the two successive blocks preceding the block with
  only two nodes. In these two blocks we choose five nodes as depicted
  in Fig.~\ref{fig:upperBound}. This means that we have $2n/3$ nodes
  plus one additional dominating node.
\end{proof}

\begin{figure}
\begin{center} 
\vspace{-14mm}
\subfigure[{Case $n \equiv 0,3[6]$}]{\includegraphics[scale=0.43]{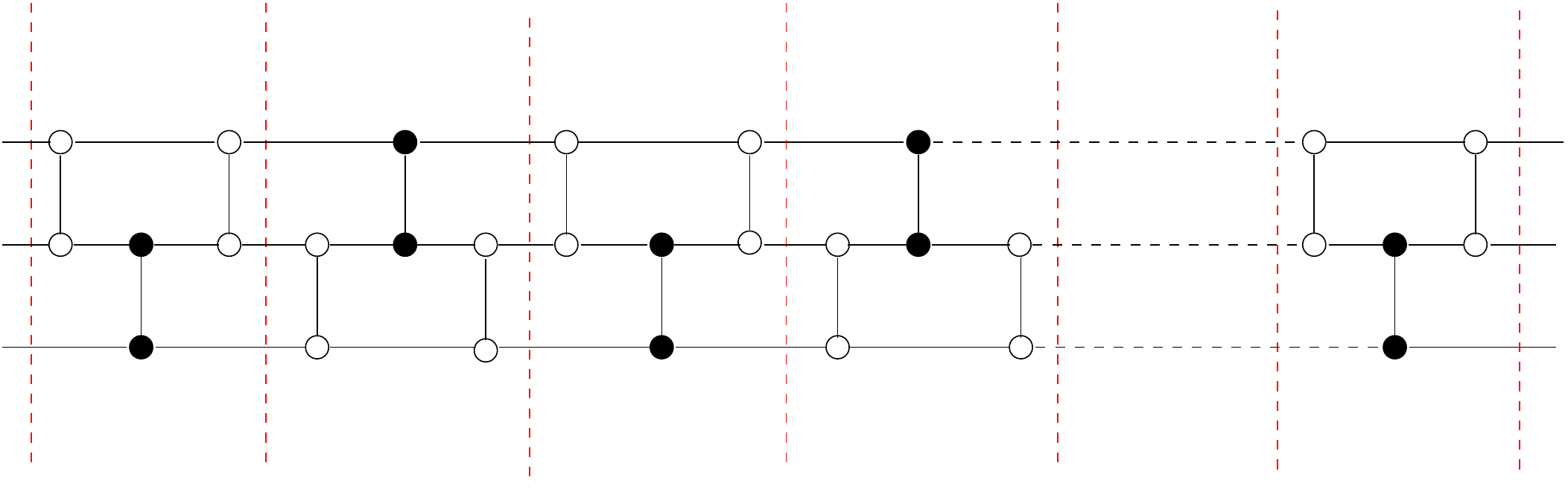}}
\subfigure[{Case $n \equiv 2[6]$}]{\includegraphics[scale=0.45]{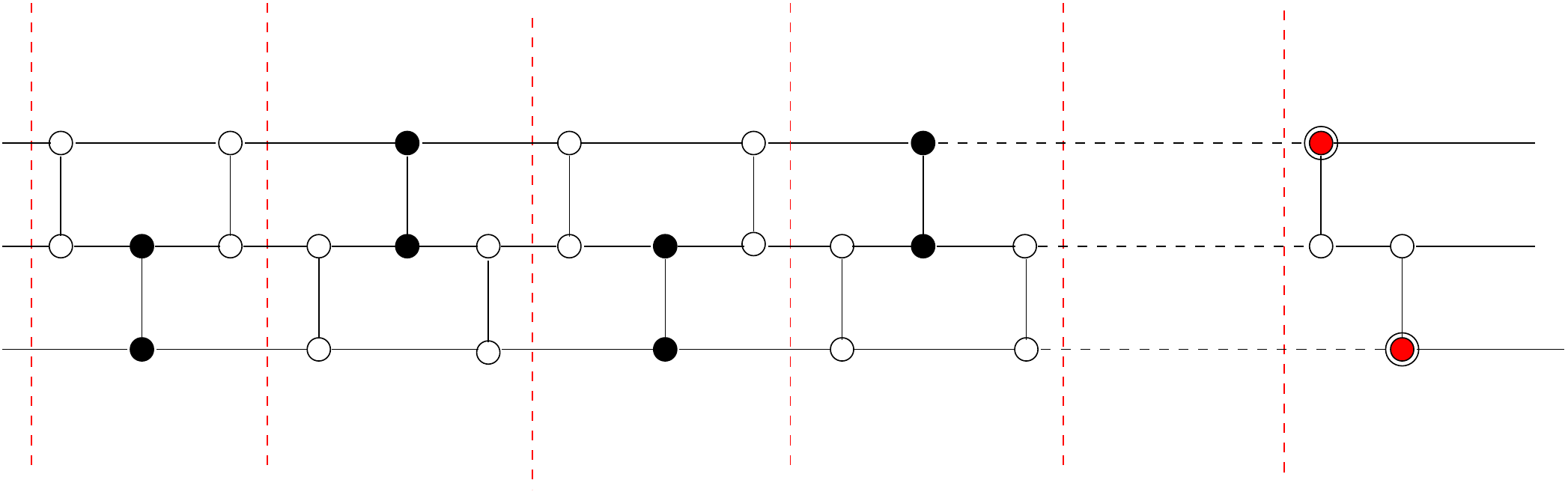}}
\subfigure[{Case $n \equiv 4[6]$}]{\includegraphics[scale=0.42]{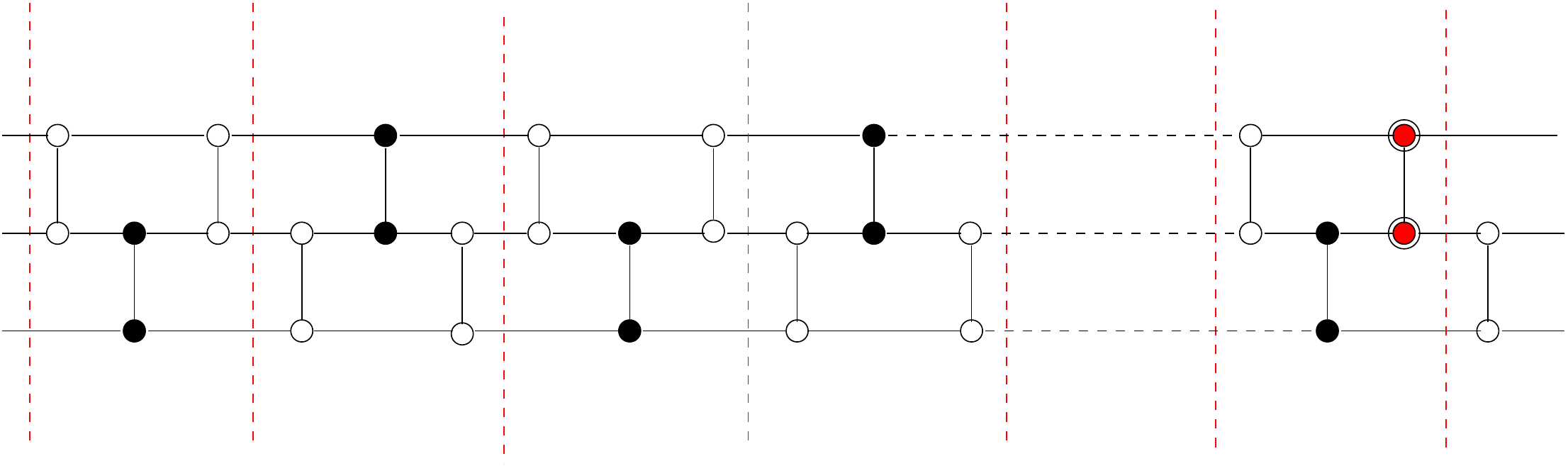}}
\subfigure[{Case $n \equiv 5[6]$}]{\includegraphics[scale=0.4]{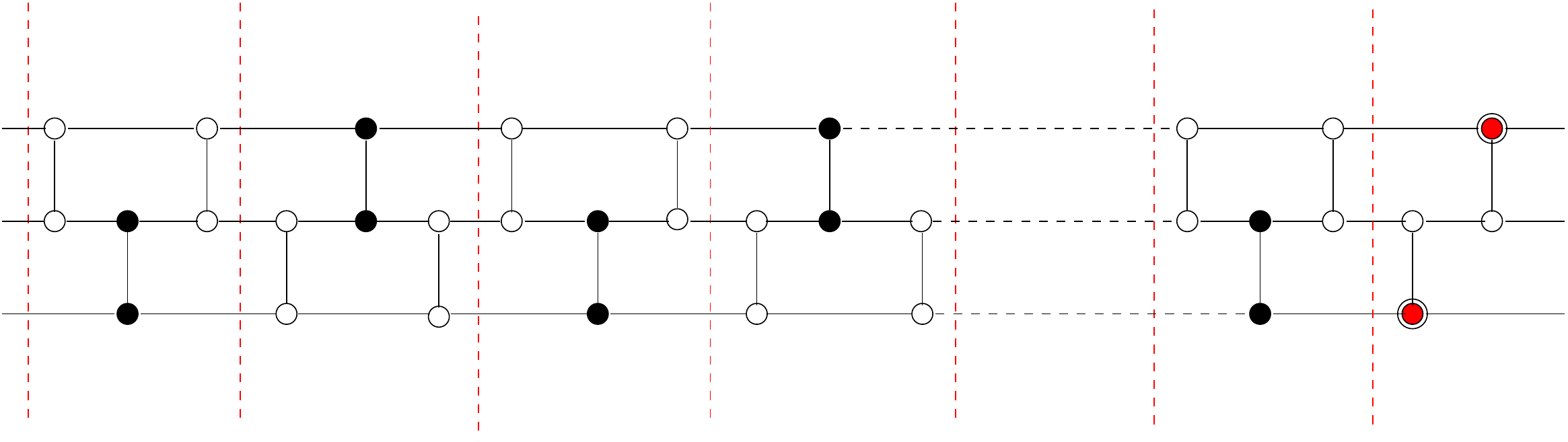}}
\subfigure[{Case $n \equiv 1[6]$}]{\includegraphics[scale=0.38]{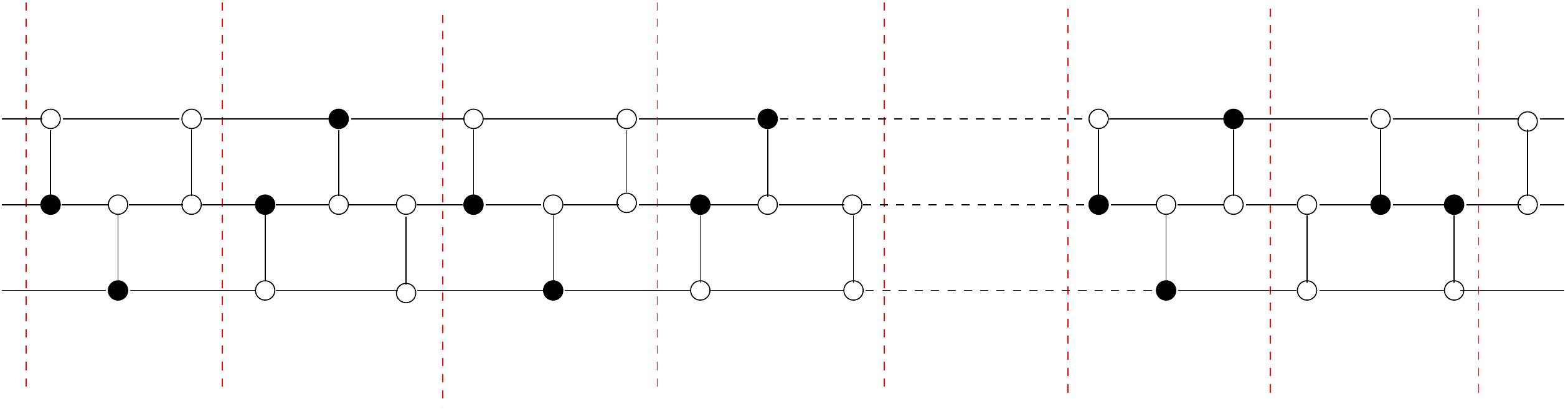}}
\end{center}
\vspace{-6mm}\caption{$f(n)$ is an upper bound of $\OTDN(P(n,2))$ for all $n > 12$}
\label{fig:upperBound}
\end{figure}

Thus, it suffices to prove that $f(n)$ is a lower bound. Assume that
there exists a minimal $\OT$-dominating set $S$ of $P(n,2)$ with
$|S|<f(n)$. Lemma~\ref{lem:smallCases} yields $n\ge 12$. The remaining
proof is split into two parts depending on whether ${\cal B}_1(S)$ is
empty or not.

\subsection{Case ${\cal B}_1(S) = \emptyset$}
\label{sec:case-cal-b_1s}
The vertices of $P(n,2)$ are grouped into $n$ pairs $p_i=\{v_i,u_i\}$ as
depicted in Fig.~\ref{fig:pairs}. Since ${\cal B}_1(S)= \emptyset$
this means that for $i=1,...,n$
 \[\gamma_S(p_i) + \gamma_{S}(p_{i+1}) + \gamma_S(p_{i+2}) \geq
2\]  (subscripts are always taken modulo $n$). Note
that $\gamma_S(p_i) \leq 2$ for all $i$. Consider the following system
of inequalities for integer valued variables $x_0,\ldots,x_{n-1}$.
  \begin{equation}
\begin{array}{rll}
 x_i & \leq&  2   \\
x_i + x_{i+1} + x_{i+2} &\ge&  2 \\
\sum_{i=0}^{n-1} x_i &<& f(n)
\end{array}\label{basiceq}
\end{equation} 

\begin{figure}%[h]
\centering
\includegraphics[scale=0.5]{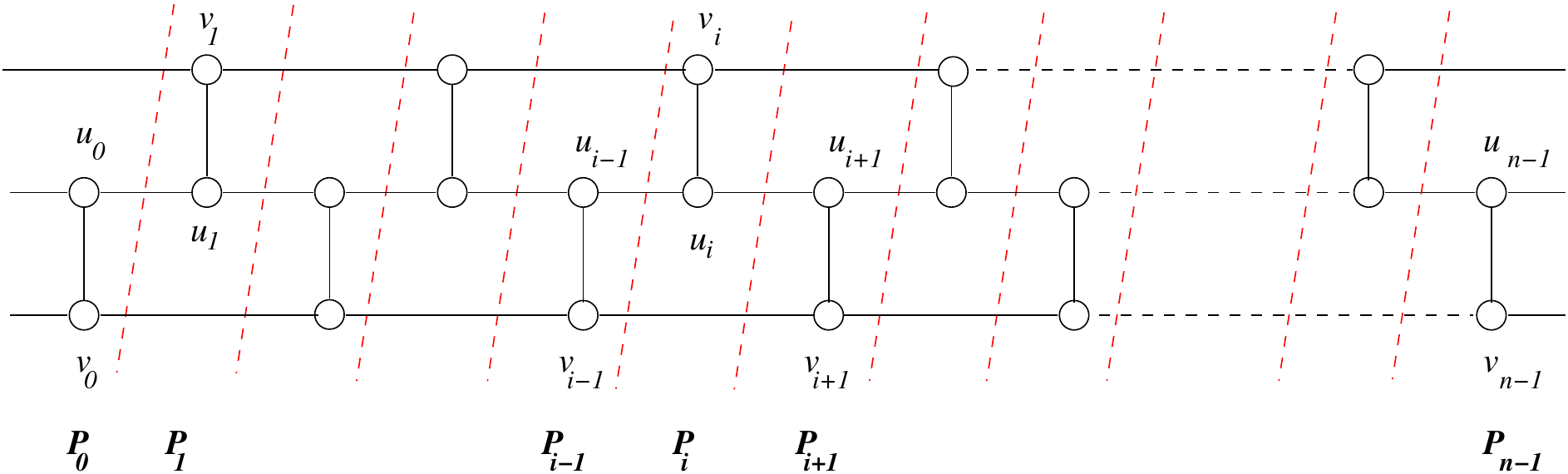}
\caption{The partition of $P(n,2)$ into $n$ pairs.}
 \label{fig:pairs}
\end{figure}

Note that $x_i=\gamma_S(p_i)$ is a solution for these equations. We
will show that no solution of Eq.~(\ref{basiceq}) is induced by a
$\OT$-dominating set. 

\begin{lem}\label{lem:basics2}
  Let $x$ be a solution of Eq.~(\ref{basiceq}) with $x_i=2$ for some
  $i$. Let $\hat{x}=x$ except $\hat{x}_{i+1}=\hat{x}_{i+2}=0$ and
  $\hat{x}_{i+3}=2$. Then $\hat{x}$ is a solution of
  Eq.~(\ref{basiceq}) with $\sum_{i=0}^{n-1} \hat{x}_i \le
  \sum_{i=0}^{n-1} x_i$.
\end{lem}
\begin{proof}
  Obviously $\hat{x}$ satisfies the first two sets of inequalities.
  Note that $x_{i+1}+x_{i+2}+x_{i+3} \ge 2$ since $x$ is a solution of
  Eq.~(\ref{basiceq}). Thus,
  $\hat{x}_{i+1}+\hat{x}_{i+2}+\hat{x}_{i+3} \le
  x_{i+1}+x_{i+2}+x_{i+3}$.
\end{proof}

  % \begin{enumerate}
  % \item If $x_{i}=x_{i+2}=2$ then $x_{i+1}=0$
  % \item If $x_{i}=x_{i+3}=2$ then $x_{i+1}=x_{i+2}=0$
  % \item If $x_{i}=2$ and $x_{i+1}=2$ or $x_j <2$ for $j=i+1,i+2,i+3$
  %   then there exist a solution $\hat{x}$ with $\sum_{i=0}^{n-1}
  %   x_i=\sum_{i=0}^{n-1} \hat{x}_i$ which coincides with $x$ except
  %   $\hat{x}_{i+1}=\hat{x}_{i+2}=0$ and $\hat{x}_{i+3}=2$.
  % \end{enumerate}

\begin{lem}\label{lem:case2}
  Let $x$ be any  solution of Eq.~(\ref{basiceq}). Then
  $x_i\le 1$ for $i = 0,\ldots, n-1$.
\end{lem}

\begin{proof}
  Let $x$ be any solution of Eq.~(\ref{basiceq}) such that $x_i=2$ for
  some $i$. Without loss of generality $i=0$. By
  Lemma~\ref{lem:basics2} there exist a solution which coincides with
  $x$ except $x_{1}=x_{2}=0$ and $x_{3}=2$. Repeatedly applying
  Lemma~\ref{lem:basics2} proves that there exits a solution $\hat{x}$
  of Eq.~(\ref{basiceq}) with $\hat{x}_{k}=2$ and
  $\hat{x}_{k+1}=\hat{x}_{k+2}=0$ for $k=0,1,\ldots, \lfloor n/3
  \rfloor$. If $n \equiv 0\,[3]$ then $\sum_{i=0}^{n-1} \hat{x}_i
  =2n/3 = f(n)$, which is impossible. Suppose $n \equiv 1\,[3]$. Then
  $\hat{x}_{n-1}=2$ otherwise the second constraint for $i=n-2$ would
  be violated. This leads to the contradiction $\sum_{i=0}^{n-1}
  \hat{x}_i =2\lfloor n/3 \rfloor +2 \ge f(n)$. Hence, $n \equiv
  2\,[3]$. Then $\hat{x}_{n-2}=2$ otherwise the second constraint for
  $i=n-2$ is not satisfied. Again this leads to the contradiction
  $\sum_{i=0}^{n-1} \hat{x}_i =2\lfloor n/3 \rfloor +2 = f(n)$. This
  proves $x_i \le 1$ for all $i$.
\end{proof}

\begin{lem}\label{lem:case1}
  If $n \not\equiv 4\,[6]$ then Eq.~(\ref{basiceq}) has no solution.
  If $n \equiv 4\,[6]$ then any solution of Eq.~(\ref{basiceq}) is a
  rotation of the solution $(1,1,0,1,1,0, \ldots, 1,1,0,1)$.
\end{lem}
\begin{proof}
  Let $x$ be any solution of Eq.~(\ref{basiceq}). By
  Lemma~\ref{lem:case2} $x_i\le 1$ for $i = 0,\ldots, n-1$. Denote by
  $n_0$ the number of variables with $x_i=0$. Thus $\sum_{i=0}^{n-1}
  x_i = n - n_0$. Note that if $x_{i}=0$ then either $x_{i+1}=1$ or
  $x_{i-1}=1$, thus no adjacent variables have both value $0$. Denote
  by $l_1,\ldots, l_{n_0}$ the lengths of maximal sequences of
  consecutive $x_i$ with $x_i=1$. Note that $l_j \ge 2$ for all $j$.
  Then \[\sum_{i=0}^{n-1} x_i = \sum_{j=1}^{n_0} l_j = 2n_0 +
  \sum_{j=1}^{n_0} (l_j-2).\] This implies \[3 \sum_{i=0}^{n-1} x_i =
  2n + \sum_{j=1}^{n_0} (l_j-2).\] If $n \equiv 0\,[3]$ then
  $\sum_{i=0}^{n-1} x_i \ge 2n/3 = f(n)$. A contradiction. If $n
  \equiv 2\,[3]$ then again this leads to the contradiction
  $\sum_{i=0}^{n-1} x_i = 2\lfloor n/3 \rfloor + (4 + \sum_{j=1}^{n_0}
  (l_j-2))/3 \ge 2\lfloor n/3 \rfloor + 2 \ge f(n)$. Finally if $n
  \equiv 1\,[6]$ then $\sum_{i=0}^{n-1} x_i = 2\lfloor n/3 \rfloor +
  (2 + \sum_{j=1}^{n_0} (l_j-2))/3\ge 2\lfloor n/3 \rfloor + 1 =f(n)$.
  This contradiction proves that for $n \not\equiv 4\,[6]$
  Eq.~(\ref{basiceq}) has no solution.

  Let $n \equiv 4\,[6]$. Then $\sum_{i=0}^{n-1} x_i = 2\lfloor n/3
  \rfloor + (2 + \sum_{j=1}^{n_0} (l_j-2))/3 <f(n) = 2\lfloor n/3
  \rfloor + 1$ implies $3 = 2 + \sum_{j=1}^{n_0} (l_j-2)$. This yields
  that there exists $i$ such that $l_i=3$ and $l_j=2$ for all
  $j\not=i$. Thus, $x$ is a rotation of the solution $(1,1,0,1,1,0,
  \ldots, 1,1,0,1)$.
\end{proof}

\begin{lem}
  The solution $x=(1,1,0,1,1,0, \ldots, 1,1,0,1)$ is not induced by a
  $\OT$-do\-minating set of $P(n,2)$.
\end{lem}

\begin{proof}
  Assume there exists a $\OT$-do\-minating set $S$ such that
  $x_i=\gamma_S(b_i)$. Two vertices of the first two pairs must be in
  $S$. All four possibilities lead to a contradiction as shown in the
  following.

  % To prove this we do proof by contradiction. Assume that the solution
  % $(1,1,0,1,1,0, \ldots, 1,1,0,1)$ can existe. We distinguish four
  % subcases depending on the position of dominating vertices in the two
  % first $2$-blocks. In the following, we consider each case and prove
  % the contradiction. For more simplicity, the vertices in the central
  % cycle of $P(n,2)$ are in position middle, the vertices in the high
  % cycle
  % are in position top and the vertices in the down cycle are in position down.\\

Case 1. $v_{0}, u_1\in S$ (see Fig.~\ref{fig:case1}). Since $S$ is
$\OT$-do\-minating the lower vertex of the last pair $p_{n-1}$ must be
in $S$. Now the same argument implies that the middle vertex of pair
$p_3$ must be in $S$. This yields that the lower vertex of pair $p_7$
must be in $S$, otherwise the lower vertex of pair $p_5$ is not
dominated. Repeating this argument shows that the lower vertex of pair
$p_{n-3}$ must be in $S$ (note that $n \equiv 4\,[6]$). Thus, $S$ does
not dominate the middle vertex of pair $p_{n-2}$. Contradiction.

\begin{figure}[htb]
\centering
\includegraphics[scale=0.5]{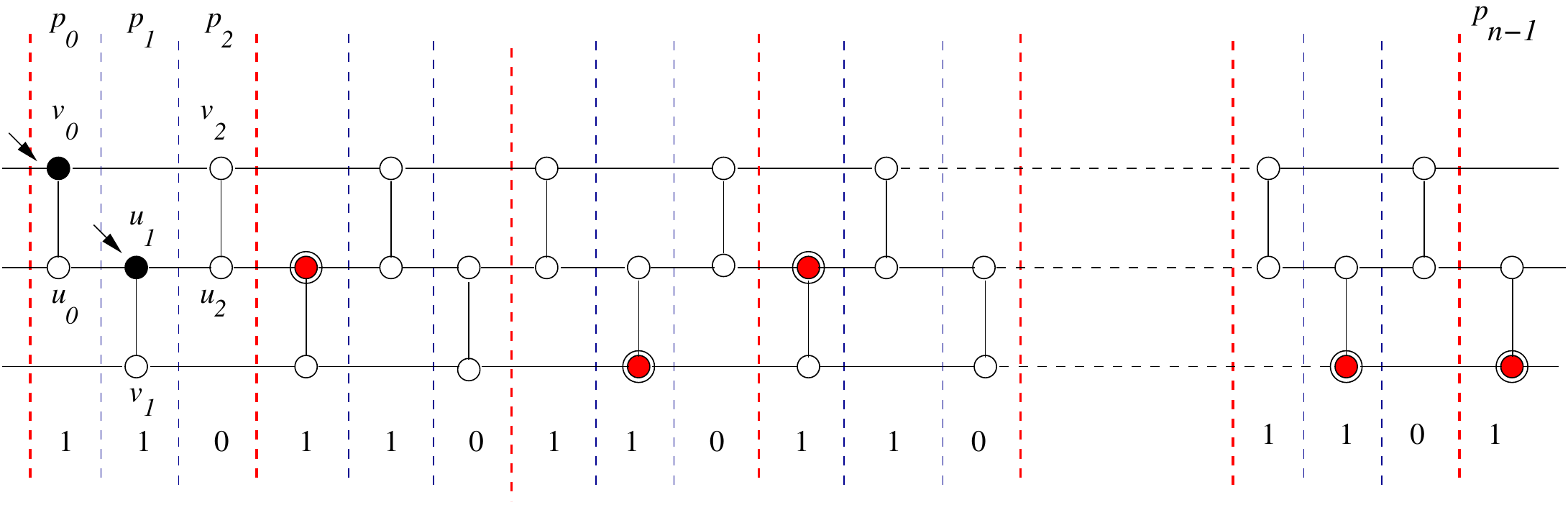}
\caption{If $v_{0}, u_1\in S$  then vertices depicted in red  must also
  be in $S$.}
 \label{fig:case1}
\end{figure}

Case 2. $u_{0}, v_1\in S$ (see Fig.~\ref{fig:case2}).
In order to dominate the middle vertex of pair $p_2$ the middle vertex
of $p_3$ must be in $S$. Similarly the lower vertex of pair $p_7$
must be in $S$ to dominate the lower vertex of $p_5$. This results in
the pattern shown in Fig.~\ref{fig:case2}. This is impossible because
all three neighbors of the lower vertex of $p_{n-1}$ are in $S$.

\begin{figure}[htb]
\centering
\includegraphics[scale=0.5]{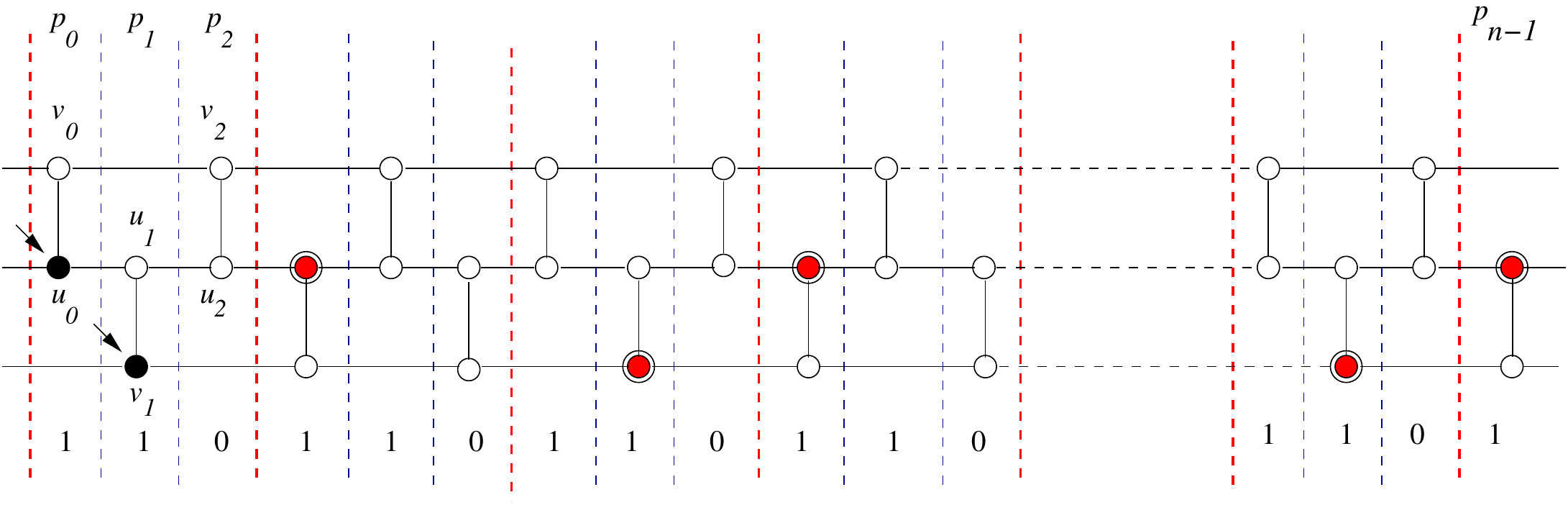}
\caption{If $u_{0}, v_1\in S$  then vertices depicted in red  must also
  be in $S$.}
 \label{fig:case2}
\end{figure}

Case 3. $u_{0}, u_1 \in S$. The same reasoning as above leads
to the situation depicted in Fig.~\ref{fig:case3}. This gives also
rise to a contradiction since the upper vertex of pair $p_{n-2}$ is
not dominated.

\begin{figure}[htb]
\centering
\includegraphics[scale=0.5]{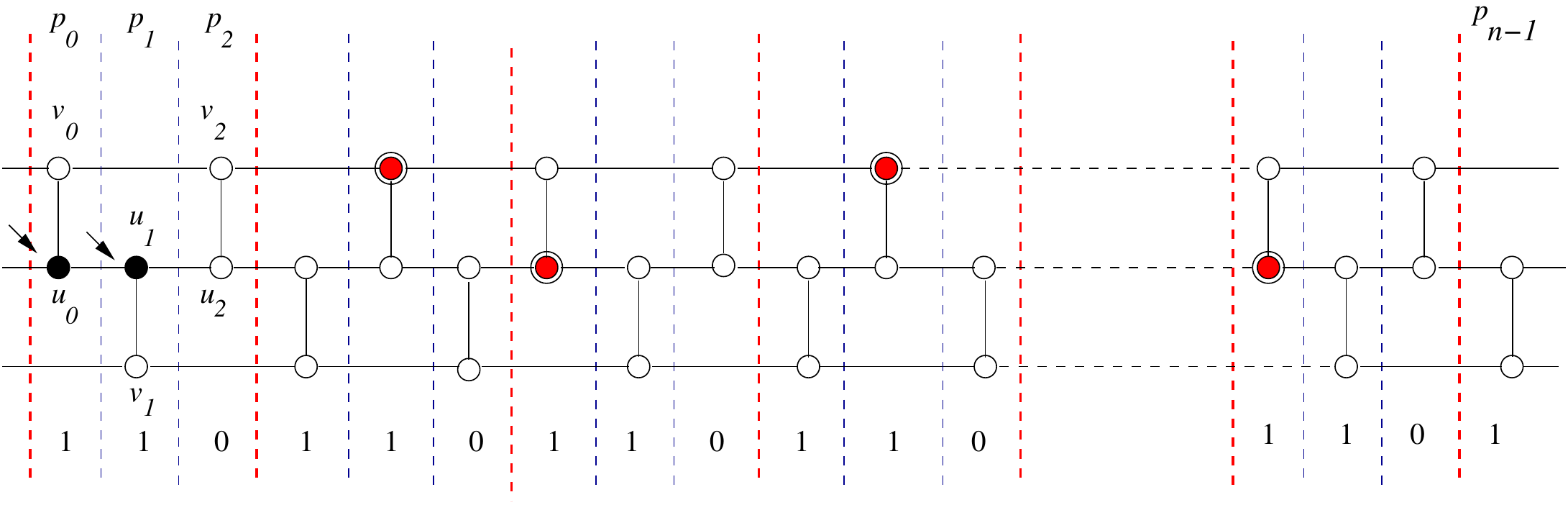}
\caption{If $u_{0}, u_1\in S$  then vertices depicted in red  must also
  be in $S$.}
 \label{fig:case3}
\end{figure}

Case 4. $v_{0}, v_1\in S$. The same reasoning as above leads
to the situation depicted in Fig.~\ref{fig:case4}. This is impossible
because all three neighbors of the lower vertex of $p_{n-1}$ are in
$S$.

\begin{figure}[htb]
\centering
\includegraphics[scale=0.5]{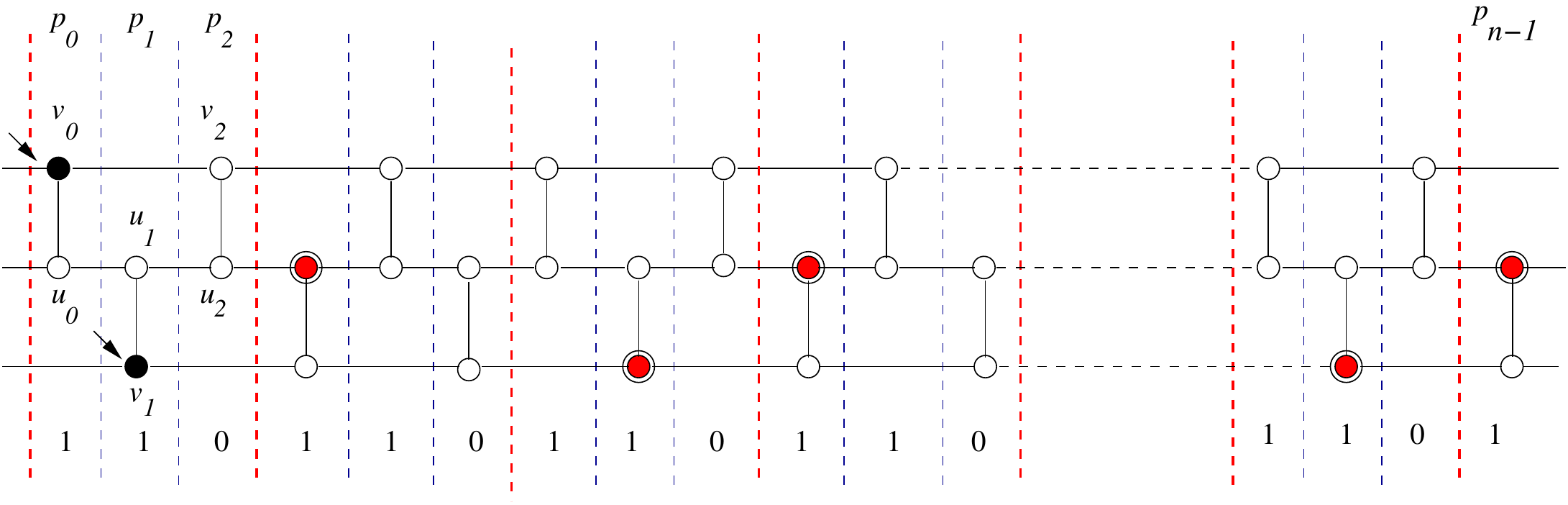}
\caption{If $v_{0}, v_1\in S$  then vertices depicted in red  must also
  be in $S$.}
 \label{fig:case4}
\end{figure}
\end{proof}

This concludes the proof of Theorem~\ref{thm:thm1} for the case ${\cal B}_1(S)
= \emptyset$.

% \begin{prop} \label{lem:withoutB1}
%   $|S| \ge f(n)$ for any $[1,2]$-dominating set $S$ of $P(n,2)$ with
%   ${\cal B}_1(S)= \emptyset$.
% \end{prop}

\subsection{Case ${\cal B}_1(S) \not= \emptyset$}
\label{sec:case-cal-not-b_1s}

The following simple observation is based on the fact that the central
vertex of a block $b$ can only be dominated by a vertex within $b$.

\begin{lem}\label{lem:fourCases}
  Any positive block $b\in {\cal B}_1(S)$ corresponds to one of the
  four blocks shown in Fig.~\ref{img:1classif}. A similar result holds
  for negative blocks.
\end{lem}

\begin{figure*}[h]
  \hfill
  \subfigure[Type A]{\includegraphics{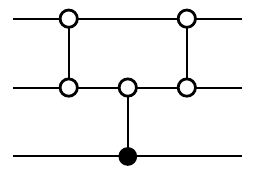}}
  \hfill
  \subfigure[Type B]{\includegraphics{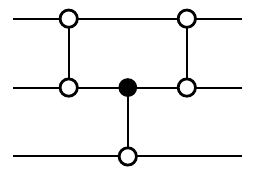}}
  \hfill
  \subfigure[Type C]{\includegraphics{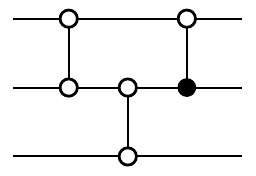}}
  \hfill
  \subfigure[Type D]{\includegraphics{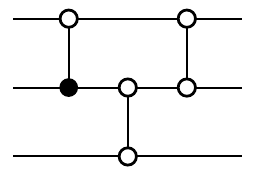}}
  \hfill\null\relax
  \caption{The four types of positive blocks with $\gamma_S(b)=1$.}
  \label{img:1classif}
\end{figure*}

In the following the four different types of blocks are considered
individually. 

\begin{lem} \label{noB} 
Let $S$ be a $\OT$-dominating set of $P (n,
  2)$ containing a block $b$ of type B and $n \geq 12$. Then there
  exists a $\OT$-dominating set $S'$ of $P (n, 2)$ not containing a
  block of type B such that $\lvert S' \rvert = \lvert S \rvert$.
\end{lem}

%\begin{figure}[htb]
%\centering
%\includegraphics[scale=0.7]{blockB}
%\caption{Block $b$ is of type B. {\bf Label blocks as $b^-$, $b$, and $b^+$!}}
% \label{fig:bloc}
%\end{figure}

\begin{proof}
  In order to dominate $v_{i-1}$ and $v_{i+1}$ from block $b$,
  vertices $v_{i-3}$ from block $b^{-}$ and $v_{i+3}$ from $b^{+}$ need to be
  in $S$. The idea is to move some
  dominating nodes such that block $b$ is not no longer of type B and
  no new block of type B emerges while $S$ is still $\OT$-dominating
  and the cardinality of $S$ remains. The proof is divided into four
  cases, considering whether $u_{i-2}$ from block $b^{-}$ and
  $u_{i+2}$ from $b^{+}$ are in $S$ or not. The notation of the nodes is taken from Fig.~\ref{img:blocks}.

  Case 1. $u_{i-2}, u_{i+2} \in S$. If $v_{i-4}$ and
  $v_{i+4}$ are not in $S$ then $S'= S / \{ u_i\} \cup \{ v_i\}$. If
  $v_{i-4}$ or $v_{i+4}$ are in $S$ then $S'= S / \{ u_{i-2} \} \cup
  \{ u_{i-1}\}$ or $S'= S / \{ u_{i+2} \} \cup \{ u_{i+1}\}$.

  Case 2. $u_{i-2}\not\in S$, $u_{i+2}\in S$. To dominate
  $u_{i-2}$ and $v_{i-2}$ we consider two subcases.

  Subcase 2.1. $v_{i-2}\in S$. If $u_{i-3}$ is not in $S$ then
  $S'= S / \{ u_{i}\} \cup \{ u_{i-1}\}$. If $u_{i-3}\in S$ then there
  are three possibilities depending on which vertex dominates
  $u_{i+4}$. Hence, if $u_{i+3} \in S$ then $S'= S / \{ u_{i+2}\}
  \cup \{ v_{i}\}$. If $v_{i+4}\in S$ then $S'= S / \{ u_{i+2}\}
  \cup \{ u_{i+1}\}$. Otherwise, the vertex $u_{i+4}$ is dominated by
  node $u_{i+5}$ of block $b^{++}$ then $S'= S / \{ u_{i}\} \cup \{
  v_{i-1}\}$.

  Subcase 2.2. $v_{i-2} \not\in S$. This implies that $u_{i-3}$
  and $v_{i-4}$ from are both in $S$. Then $S'= S / \{ u_{i-3}\} \cup
  \{ u_{i-1}\}$.

  Case 3. $u_{i+2}\not\in S$,
  $u_{i-2}\in S$. This case is symmetric to case 2.

  Case.4. $u_{i+2},u_{i-2}\not\in S$. In order to dominate
  $u_{i-2}$ and $v_{i-2}$ two situations must be considered.
 
  Subcase 4.1. $v_{i-2}\in S$. Since $v_{i-2}$ is in $S$ and
  $v_i$ is not in $S$ then $v_{i+2}$ cannot be a dominating node. This
  yields that $u_{i+3}$ and $v_{i+4}$ are in $S$. Then $S'= S / \{
  u_{i+3}\} \cup \{ u_{i+1}\}$.

  Subcase 4.2. $v_{i-2}\not\in S$. This implies $u_{i-3},
  v_{i-4}\in S$. Therefore, $S'= S / \{ u_{i-3}\} \cup \{ u_{i-1}\}$.
 \end{proof}

The next Lemma finally completes the proof of Theorem~\ref{thm:thm1}.
 
\begin{lem}
  If ${\cal B}_1(S) \not= \emptyset$ and $n \geq 6$ then $|S|\ge
  f(n)$.
\end{lem}
\begin{proof}
  Let $n$ be minimal such that the lemma is false. Then $n \ge 12$
  by Lemma~\ref{lem:smallCases}. Let ${{\cal S}}_B$ the set of all
  $\OT$-dominating sets $S$ of $P(n,2)$ not containing a block of
  type B and $|S|< f(n)$. Then ${\cal B}_1(S)\not=\emptyset$ for all
  $S\in {{\cal S}}_B$ by the first part of the proof. Let $p$ be the
  largest number such that $|{\cal B}_1(S)|\ge p$ for each $S\in
  {{\cal S}}_B$. Then $p\ge 1$. Let $\M$ be the set of all $S\in
  {{\cal S}}_B$ with $|{\cal B}_1(S)|=p$.

  % Let $S\in \M$. Suppose there exists a positive block $B_i\in {\cal
  %   B}_1(S)$ of type D such that vertex $v_{i-3}$ is not in $S$. Then
  % the three left most nodes $\{v_{i+2},u_{i+2},v_{i+3}\}$ of the
  % negative 6-block to the right of $B_i$ must be in $S$. Then
  % $S'=S\setminus\{u_{i+2}\}\cup \{u_{i+1}\}$ is a $[1,2]$-dominating
  % set of $P(n,2)$ with $|{\cal B}_1(S')|<|{\cal B}_1(S)|$ (see
  % Fig.\ref{img:shiftLeft1}). This contradicts the choice of $p$ and
  % proves that for any occurrence of a 6-block of type D vertex
  % $v_{i-3}$ is in $S$. Similarly, for any occurrence of a positive
  % 6-block of type C vertex $v_{i+3}$ is in $S$. The negative 6-blocks
  % of type C and $D$ have similar properties.

  % {\bf PICTURE 1: Show replacement of nodes of S.}

  {\bf Claim 1:} $P(n,2)$ does not contain a
  block of type A for any $S\in \M$.\\
  Assume false. Let $S\in \M$ and $b$ a positive block of type A. Then
  the nodes $v_{i+3}$ and $u_{i+2}$ of $b^+$ must be dominating.
  Assume $\gamma_{b^+}(S)\ge 3$. Then $S'=S\setminus\{u_{i+2}\}\cup
  \{u_{i}\}$ is also a $[1,2]$-dominating set. %(see Fig.~\ref{img:shiftLeft2A}). 
  Thus, $\gamma_{b}(S')= 2$. Then $|{\cal
    B}_1(S')| = |{\cal B}_1(S)| -1 < p$ since $\gamma_{b}(S)= 1$. This
  yields $S' \not\in {{\cal S}}_B$ and therefore ${\cal B}_1(S) =
  \emptyset$. Thus, $\gamma_{b^+}(S)=2$.

  % {\bf PICTURE 1: Show replacement of nodes of S.}

  Let $b^{++}$ be the positive block to the right of $b^+$. Then the
  nodes $u_{i+5}$ and $v_{i+6}$ of $b^{++}$ must be dominating. 
  Next we remove the nodes of the blocks
  $b$ and $b^+$ and connect the corresponding nodes of blocks $b^-$
  and $b^{++}$. The resulting graph is isomorphic to $P(n-6,2)$.
  Furthermore, $S'=S\setminus\{v_i, u_{i+2}, v_{i+3}, u_{i+5}\}$ is a
  $[1,2]$-dominating set of this graph. 
  Thus, $|S'|=|S|-4 \ge f(n-6)$ by the choice
  of $n$. Therefore $|S| \ge f(n-6) + 4 = f(n)$. This implies $|S|\ge
  f(n)$. This contradiction proves claim 1 for positive blocks of
  type A. The same argument shows that there are no negative blocks
  of type A.

  {\bf Claim 2:} $P(n,2)$ does not contain a block of type
  D for any $S\in \M$\\
  Assume false. As above we only need to consider the positive case.
  Let $S\in \M$ and $b$ a positive block of type D. Then nodes
  $v_{i+3}$ and $u_{i+2}$ of $b^+$ must be dominating. Assume
  $\gamma_{b^+}(S)=2$. Then again the nodes
  $u_{i+5}$ and $v_{i+6}$ of block $b^{++}$ must be dominating. We
  distinguish two cases. If $v_{i-2}$ is not a dominating node then
  $S'=S\setminus\{v_{i+3}\}\cup \{v_{i+1}\}$
  else ($v_{i-2}$ is a dominating node) then we have again two subcases depending on $\gamma_{b^{++}}(S)$.
  If $\gamma_{b^{++}}(S)=2$ then the nodes $u_{i+8}$ and $v_{i+9}$ of
  the block to the right of $b^{++}$ must be dominating nodes. We
  remove the nodes of the blocks $b$ and $b^+$ and connect the
  corresponding nodes of blocks $b^-$ and $b^{++}$ with
  $S'=S\setminus\{u_{i-1}, u_{i+2}, v_{i+3}, v_{i+6}\}$. Similar to
  the proof of claim 1 this leads to a contradiction. If
  $\gamma_{c}(S)\geq 3$ then at least one of the nodes $v_{i+5}$ and
  $u_{i+6}$ is a dominating node. Then we again remove the nodes of
  the blocks $b$ and $b^+$ and connect the corresponding nodes of
  blocks $b^-$ and $b^{++}$ with $S'=S\setminus\{u_{i-1}, u_{i+2},
  v_{i+3}, u_{i+5}\}$. Similar to the proof of claim 1 this leads to a
  contradiction.

  Hence, $\gamma_{b^+}(S)\ge 3$. In the following we will construct a
  new $[1,2]$-dominating set $S'$ with $|{\cal B}_1(S')|< p$. This is
  a contradiction.
  
  Case 1. $v_{i+2},u_{i+3}\in S$. There are three subcases. If
  $v_{i-3}\not\in S$ then $S'=S\setminus\{u_{i+2}\}\cup \{v_{i+1}\}$
  and if $v_{i-2}\not\in S$ then $S'=S\setminus\{u_{i+2}\}\cup
  \{u_{i}\}$. If $v_{i-3},v_{i-2}\in S$
  then $S'=S\setminus\{u_{i-1}, u_{i+2}\}\cup \{u_{i}, v_{i} \}$.

  Case 2. Neither $v_{i+2}$ nor $u_{i+3}$ are in $S$. Since
  $\gamma_{b^+}(S)\ge 3$ this implies that $v_{i+4}$ is a dominating
  node and $S'=S\setminus\{u_{i+2}\}\cup \{u_{i+1}\}$.
  
  Case 3. If $v_{i+2}\in S$ and $u_{i+3}\not\in S$
  then $S'=S\setminus\{u_{i+2}\}\cup \{u_{i+1}\}$.
  
  Case 4. If $v_{i+2}\not\in S$ and $u_{i+3}\in S$ we distinguish two
  cases: If $v_{i+4}\in S$ then $S'=S\setminus\{u_{i+2}\}\cup
  \{u_{i+1}\}$ else we have four subcases depending on which node
  dominates $u_{i+5}$:
  \begin{enumerate}
  \item If $v_{i+5}\in S$ then
    $S'=S\setminus\{v_{i+3}\}\cup \{u_{i+1}\}$.
  \item If $u_{i+5}\in S$ then
    $S'=S\setminus\{u_{i+3}\}\cup \{u_{i+1}\}$.
  \item If $u_{i+6}\in S$ then we distinguish three cases depending on
    which node dominates $v_{i+4}$. If $u_{i+4}\in S$ then
    $S'=S\setminus\{u_{i+2}, u_{i+4} \}\cup \{v_{i+2}, u_i\}$. If
    $v_{i+4}\in S$ then $S'=S\setminus\{u_{i+2}, v_{i+4} \}\cup
    \{v_{i+2}, u_i\}$. Finally if $v_{i+6}\in S$ then we
    remove the nodes of the blocks $b$ and $b^+$ and connect the
    corresponding nodes of blocks $b^-$ and $c$.
  \item If $u_{i+4}\in S$ then $S'=S\setminus\{u_{i+2},u_{i+3}\}\cup
    \{u_{i+1}, v_{i+4}\}$.
  \end{enumerate}
  This proves claim 2.
  
%  If $v_{i+2}$ is dominating then
%  $S'=S\setminus\{u_{i+2}\}\cup \{u_{i}\}$ is a $[1,2]$-dominating set
%  of $P(n,2)$ (see Fig.~\ref{img:shiftLeft4A}). This implies $|{\cal
%    B}_1(S')|< p$. Thus, $v_{i+2}$ is not dominating.
%  If $v_{i+4}$ is dominating then $S'=S\setminus\{u_{i+2}\}\cup
%  \{u_{i}\}$ is a $[1,2]$-dominating set of $P(n,2)$.
%  Thus, $v_{i+4}$ is also not dominating. Next suppose that $u_{i+3}$
%  and $u_{i+4}$ are both dominating. We need to distinguish the cases
%  that $u_{i+5}$ is dominating or not. Let
%  $S'=S\setminus\{u_{i+2},u_{i+3}\}\cup \{u_{i},v_{i+4}\}$ resp.\
%  $S'=S\setminus\{u_{i+2},u_{i+3}\}\cup \{u_{i},v_{i+2}\}$. In both
%  cases $S'$ is a $[1,2]$-dominating set of $P(n,2)$ (see
%  Fig.~\ref{img:shiftLeft5A}). This contradiction leaves the case that
%  $u_{i+3}$ is dominating (because $S$ is $[1,2]$-dominating the case
%  that only $u_{i+4}$ is dominating, is impossible). Then again
%  $S'=S\setminus\{u_{i+2},u_{i+3}\}\cup \{u_{i},v_{i+4}\}$ is a
%  $[1,2]$-dominating set of $P(n,2)$ (see Fig.~\ref{img:shiftLeft5A}).
%  This is impossible and proves claim 2.

  {\bf Claim 3:} $P(n,2)$ does not contain a
  block of type C for any $S\in \M$\\
  This case is symmetric to the second claim.

  {\bf Claim 4:} $P(n,2)$ does not contain a
  block of type B for any $S\in \M$\\
  If $S$ contains a block of type B then by Lemma~\ref{noB} there
  exists $S'\in {\cal S}_B$ which does not contain a block of type B.
  The above claims yield ${\cal B}_1(S')=\emptyset$. This
  contradiction concludes the proof of the lemma.
\end{proof}

\section{Determination of $\OTTDN(P(n,2))$}
In this section, we analyze the $[1,2]$-total dominating sets of
$P(n,2)$ and prove the Theorem~\ref{thm:thm2}. For the case $n=5$ we
refer to Fig.~\ref{fig:graph56}. We split the proof into two lemmata.
Denote by $g(n)$ the value of the right side of the equation in
Theorem \ref{thm:thm2}.

\begin{lem}
  $\OTTDN(P(n,2)) \le g(n)$ for $n>5$.
\end{lem}

\begin{proof}
  In Fig.~\ref{constructotal}, we give the construction of the minimum
  $[1,2]$-total dominating set in $P(n,2)$ for $n \equiv 1[6]$. The
  proposed construction is based on the selection of one pair of nodes
  of the middle in each block which corresponds to $2n/3$ nodes. Then,
  we add two additional dominating nodes as depicted in color red in
  Fig.~\ref{constructotal}. For all other cases we refer to
  Fig.~\ref{fig:upperBound} since the provided sets are already total
  dominating sets.
\end{proof}

\begin{lem}
  $\OTTDN(P(n,2)) \ge g(n)$ for $n>5$.
\end{lem}
\begin{proof}
  For $n \not\equiv 1[6]$ this follows from Theorem~\ref{thm:thm1}. It
  remains to consider the case $n \equiv 1[6]$. Let $S$ be a total
  \OT-dominating set of minimum size of $P(2,n)$ with $|S| < g(n)$.
  Let $G[S]$ be subgraph induced by $S$. By definition of a
  $[1,2]$-total dominating set, each connected component of $G[S]$ has
  at least two vertices and every vertex of $G[S]$ has degree $1$ or
  $2$. Hence, every connected component is either a path or a cycle.
  Let $x_l$ and $y_l$ be the numbers of connected components that are
  paths and cycles of order $l$, respectively. Observe that $x_1 = 0$
  and $y_1 = \dots = y_4 = 0$. Moreover, each path of order $l$
  dominates at most $2l+2$ vertices and each cycle of $l$ vertices
  dominates at most $2l$ vertices. Thus,
\begin{equation}
   \sum_{l\ge 2} (2l+2)x_l + 2ly_l \ge 2n
   \label{eq:formel1}
\end{equation}
\begin{equation}
\sum_{l\ge 2} l(x_l + y_l) = |S|
   \label{eq:formel2}
\end{equation}
From (\ref{eq:formel1}) and (\ref{eq:formel2}) we can deduce
\begin{equation}
|S| + \sum_{l\ge 2} x_l \ge n
   \label{eq:formel3}
\end{equation}
Also observe that
\begin{equation}
\sum_{l\ge 2} lx_l \ge 2 \sum_{l\ge 2} x_l
   \label{eq:formel4}
\end{equation}

Let $n = 6k+1$. Then $g(n) = 4k+2$ and $|S| < 4k+2$. Inequality
(\ref{eq:formel3}) becomes $|S| + \sum_{l\ge 2} x_l \ge 6k+1$, thus
$\sum_{l\ge 2} x_l \ge 2k$. From (\ref{eq:formel2}) and using
(\ref{eq:formel4}), we obtain $4k \le \sum_{l\ge 2} lx_l + \sum_{l\ge
  2} ly_l \le 4k+1$. This implies $\sum_{l\ge 2} ly_l = 0$, thus $4k
\le |S| = \sum_{l\ge 2} lx_l \le 4k+1$. Since $\sum_{l\ge 2} lx_l \le
4k+1$ and $\sum_{l\ge 2} x_l \ge 2k$, we have $\sum_{l\ge 2} lx_l \le
2\sum_{l\ge 2} x_l+1$. This is only possible if $x_3 = 1$ and $x_j =
0$ for all $j >3$. Thus, $G[S]$ is the union one path $P_3$ and $x_2$
paths $P_2$. Since every $P_2$-component can dominate at most $6$
vertices and the $P_3$-component can dominate at most $8$ vertices, we
deduce $6x_2 + 8 \ge 12k+2 =2n$. On the other hand, recall $|S| = 2x_2
+ 3 \le 4k+1$ thus $6x_2 + 8 \le 12k+2$. Hence, $6x_2 + 8 = 12k+2$.
This implies that $P(n,2)$ can be partitioned into $x_2$ components as
shown in Fig.~\ref{fig:components}(a) and one component shown in
Fig.~\ref{fig:components}(b). Suppose such a partitioning exists. In
the following we study the partitioning by making consecutive
extractions of components. Extracting a component means deleting all
its vertices from the graph. Moreover, an extraction is said to be
\textit{forced} if there is no other option. Recall that the set of
vertices of $P(n,2)$ is the union of the two sets $U = \{u_0,\dots,
u_{n-1}\}$ and $V = \{v_0,\dots, v_{n-1}\}$. Vertices of $U$ and $V$
form the two main cycles of $P(n,2)$ respectively. Either all three
vertices of the $P_3$-component are on the same main cycle or two of
them are on one cycle and the third on the other. In the first case,
once the $P_3$ dominated component is extracted, the next forced
extraction of a $P_2$ dominated component would imply the appearance
of a vertex with a degree $2$ (see Fig.~\ref{fig:comp}(a)). In the
second case, after extracting the $P_3$ dominated component and after
several forced extractions of $P_2$ dominated components (see
Fig.~\ref{fig:comp}(b)), it becomes obvious that such a partitioning
is impossible. Hence $x_3 = 0$, a contradiction.
\end{proof}

\begin{figure}[htb]
\begin{center}
\subfigure[]{ \includegraphics[height=1.7cm]{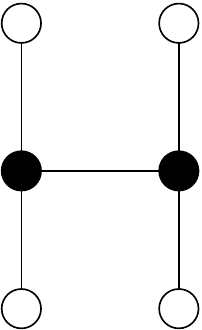}}\label{subfig:HP2} \hspace*{2cm}
\subfigure[]{\includegraphics[height=1.7cm]{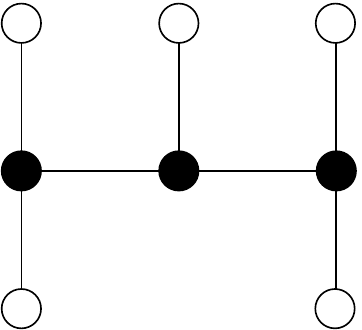}}\label{subfig:HP3}
 
\end{center}

\caption{Maximal components induced by $P_2$ and $P_3$ and their neighbors.}
\label{fig:components}
\end{figure}

\begin{figure}[!ht]
\begin{center}
\subfigure[]{\includegraphics[scale=0.4]{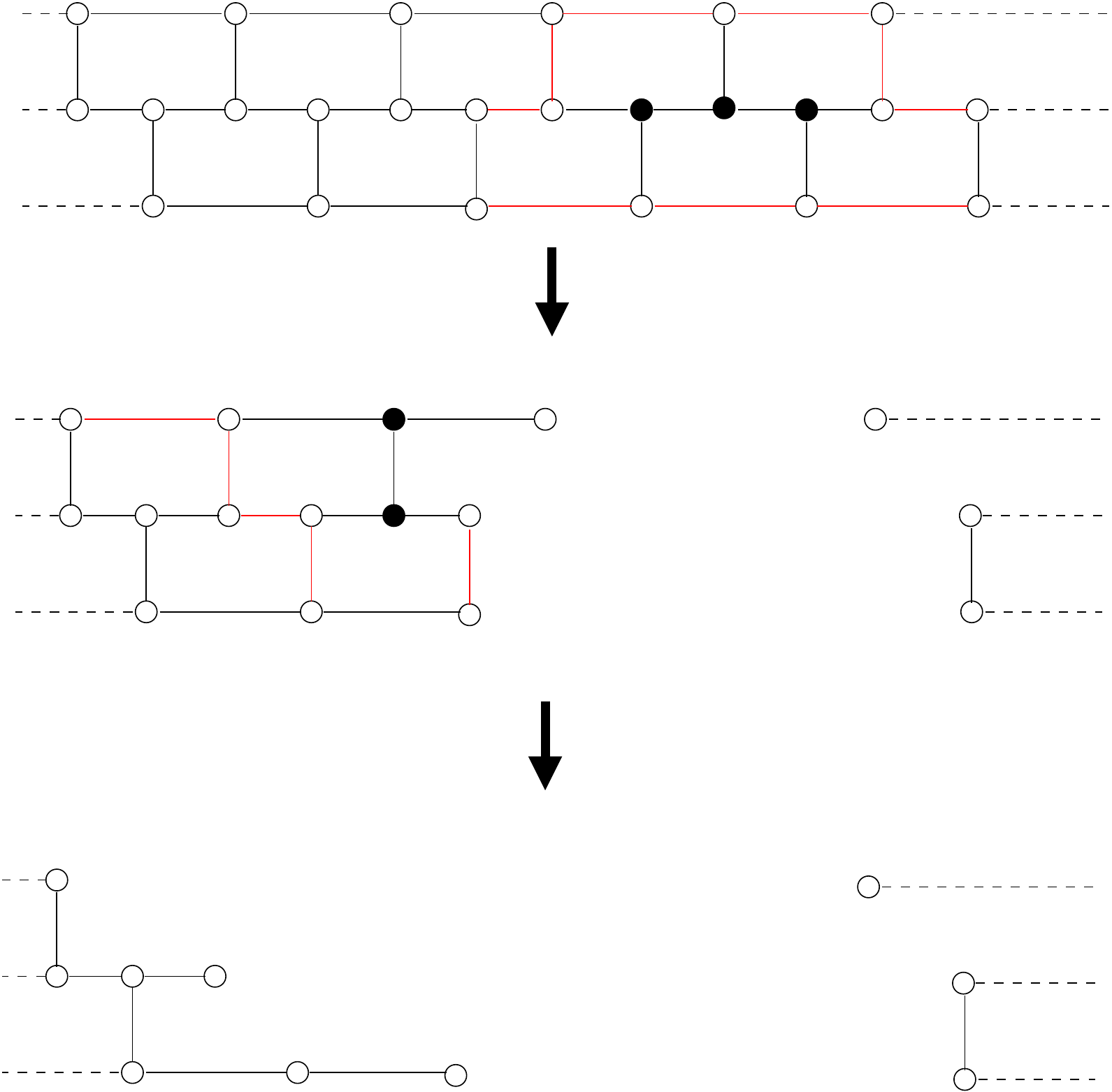}}\label{subfig:P2}
\quad \quad
\subfigure[]{\includegraphics[scale=0.4]{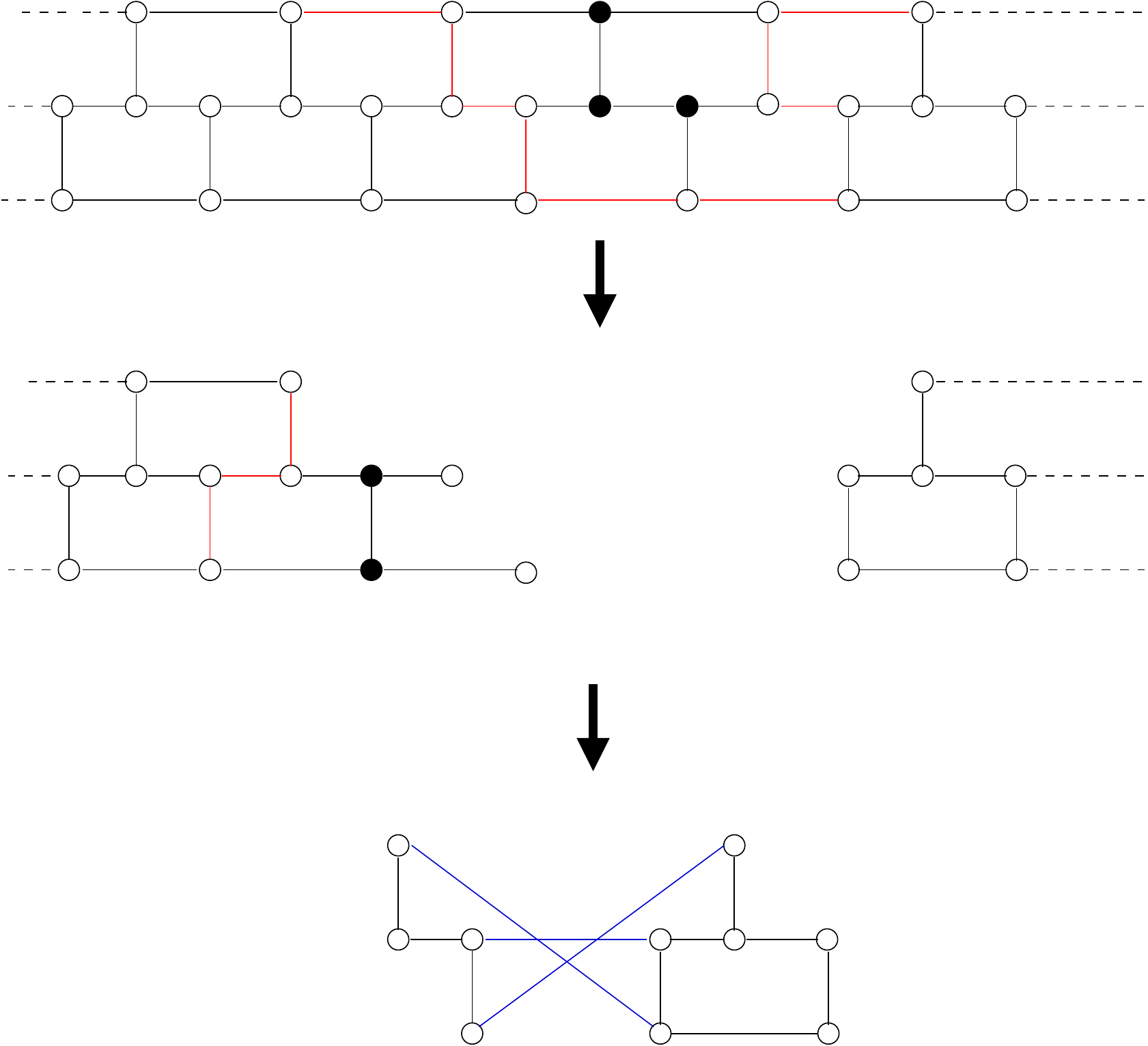}}\label{subfig:P3}
\caption{Impossible partitionings.}
\label{fig:comp}
\end{center}
\end{figure}

% In the following we study $\OTTDN(P(n,2))$ according to
% residue of $n[6]$.
% \begin{enumerate}[(a)]
% \item $n = 5$: This case can be checked by inspection.
% \item $n \equiv 0[3]$: If $n = 3k$ then $\OTTDN(P(n,2)) = 2k$ and $|S| < 2k$.
% Inequality (3) becomes $|S| + \sum_{l\ge 2} x_l \ge 3k$, thus $\sum_{l\ge 2} x_l \ge k+1$. 
% From (2), we deduce $\sum_{l\ge 2} l(x_l + y_l) < 2k$ thus $\sum_{l\ge 2} lx_l + \sum_{l\ge 2} ly_l < 2k$. 
% Using (4), we obtain $2k+ 2+ \sum_{l\ge 2} ly_l < 2k$, a contradiction.
% \item $n \equiv 2,4,5[6]$: If $n = 6k+i$ with $i \in \{2,4,5\}$ then $\OTTDN(P(n,2)) = 4k+2$ and $|S| < 4k+2$.
% Inequality (3) becomes $|S| + \sum_{l\ge 2} x_l \ge 6k+i - 4k-2$ for $i \in \{2,4,5\}$, thus 
% $\sum_{l\ge 2} x_l \ge 2k+j$  with $j \in \{1,3,4\}$. 
% From (2), we deduce $\sum_{l\ge 2} l(x_l + y_l) < 4k+2$ thus 
% $\sum_{l\ge 2} lx_l + \sum_{l\ge 2} ly_l < 4k+2$. 
% Using (4), we obtain $4k+2j + \sum_{l\ge 2} ly_l < 4k+2$ with $j \in \{1,3,4\}$, a contradiction.

% \end{enumerate}

\begin{figure}
\begin{center} 
\includegraphics[scale=0.38]{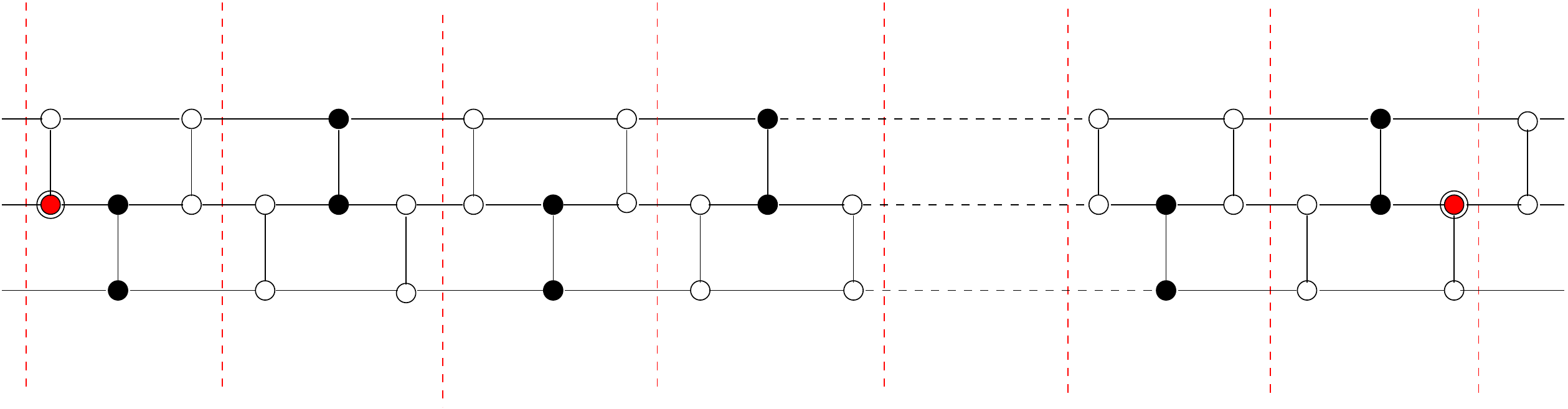}
\end{center}

\caption{The construction of $\OTTDN(P(n,2))$ for $n \equiv 1[6]$.}
\label{constructotal}
\end{figure}

\section{Conclusion}
\label{sec:conclusion}

Generalized Petersen graphs are very important structures in computer science and communication techniques since 
their particular structures and interesting properties.
In this paper, we considered a variant of the dominating set problem, called the $[1,2]$-dominating set problem.
We studied this problem in generalized Petersen graphs $P(n,k)$ for $k=2$. 
We gave the exact values of the $[1,2]$-domination numbers
and the $[1,2]$-total domination numbers of $P(n,2)$.
Obviously $\OTDN(P(n,1)) = \gamma(P(n,1))$ and so as future work we suggest to study 
the $[1,2]$-domination numbers of $P(n,k)$ with $k \geq 3$.

%\section*{References}

\bibliographystyle{plain}
\bibliography{document}

\end{document}